\documentclass{article}

\usepackage[cmex10]{amsmath}
\usepackage{amssymb}
\usepackage{amsthm}
\usepackage{enumerate}
\usepackage{amscd}

\theoremstyle{plain}
\newtheorem{theorem}{Theorem}
\newtheorem{lemma}[theorem]{Lemma}
\newtheorem{proposition}[theorem]{Proposition}
\newtheorem{corollary}[theorem]{Corollary}

\theoremstyle{definition}
\newtheorem{remark}[theorem]{Remark}

\usepackage{url}

\newcommand{\beq}{\begin{equation*}}
\newcommand{\eeq}{\end{equation*}}

\DeclareMathOperator{\ev}{ev}

\DeclareMathOperator{\Sym}{Sym}
\DeclareMathOperator{\Tr}{Tr}

\DeclareMathOperator{\dmin}{d_{min}}
\DeclareMathOperator{\ddmin}{d_{min}^{\langle 2\rangle}}
\DeclareMathOperator{\djmin}{d_{min}^{\langle 1+q^j\rangle}}
\DeclareMathOperator{\dsmin}{d_{min}^{\langle 1+q^s\rangle}}
\DeclareMathOperator{\dtmin}{d_{min}^{\langle t\rangle}}
\DeclareMathOperator{\dumin}{d_{min}^{\langle t+1\rangle}}
\DeclareMathOperator{\drel}{\delta}
\DeclareMathOperator{\ddrel}{\delta^{\langle 2\rangle}}
\DeclareMathOperator{\dsrel}{\delta^{\langle 1+q^s\rangle}}
\DeclareMathOperator{\dtrel}{\delta^{\langle t\rangle}}
\DeclareMathOperator{\rate}{R}
\DeclareMathOperator{\rated}{R^{\langle 2\rangle}}
\DeclareMathOperator{\ratet}{R^{\langle t\rangle}}
\DeclareMathOperator{\dimt}{dim^{\langle t\rangle}}
\DeclareMathOperator{\dimu}{dim^{\langle t+1\rangle}}

\newcommand{\Fq}{{\mathbb{F}_q}}
\newcommand{\Fqr}{{\mathbb{F}_{q^r}}}
\newcommand{\Fqrs}{{\mathbb{F}_{q^{2s+1}}}}
\newcommand{\N}{\mathbb{N}}

\newcommand{\F}{\mathbb{F}}
\newcommand{\SqF}{S^2_{\Fq}}
\newcommand{\ie}{\emph{i.e. }}
\newcommand{\eg}{\emph{e.g. }}
\newcommand{\longto}{\longrightarrow}
\newcommand{\tens}{\otimes}
\newcommand{\moins}{\setminus}
\newcommand{\cA}{\mathcal{A}}
\newcommand{\cB}{\mathcal{B}}
\newcommand{\cF}{\mathcal{F}}

\renewcommand{\moins}{\setminus}

\def\epsilon{\varepsilon}

\begin{document}

\title{Asymptotically good binary linear codes with asymptotically good self-intersection spans}

%\author{Hugues Randriambololona\\
%ENST, Paris (``Telecom ParisTech'')}
\author{Hugues Randriambololona}

\maketitle

\begin{abstract}
If $C$ is a binary linear code, let $C^{\langle 2\rangle}$ be the linear
code spanned by intersections of pairs of codewords of $C$.
We construct an asymptotically good family of binary linear codes
such that, for $C$ ranging in this family, the $C^{\langle 2\rangle}$ also form
an asymptotically good family. For this we use algebraic-geometry codes,
concatenation, and a fair amount of bilinear algebra.

More precisely, the two main ingredients used in our construction are,
first, a description of the symmetric square of an odd degree extension
field in terms only of field operations of small degree,
and second, a recent result of Garcia-Stichtenoth-Bassa-Beelen
on the number of points of curves on such an odd degree extension field.
\end{abstract}

\section{Statement of result}
\label{Intro}

Let $q$ be a prime power,
%(\eg $q=2$),
and $\Fq$ the field with $q$ elements.
For any integer $n\geq1$, let $*$ denote coordinatewise multiplication
in the vector space $(\Fq)^n$, so
\beq
(x_1,\dots,x_n)*(y_1,\dots,y_n)=(x_1y_1,\dots,x_ny_n).
\eeq
For $C\subset(\Fq)^n$ a linear subspace,
\ie a $q$-ary linear code of length $n$,
let
\beq
C*C=\{c*c'\;|\;c,c'\in C\}\,\subset\,(\Fq)^n
\eeq
and let
\begin{equation}
\label{C<2>}
C^{\langle 2\rangle}=\langle C*C\rangle=\{\sum_{c,c'\in C}\alpha_{c,c'}c*c'\;|\;\alpha_{c,c'}\in\Fq\}
\end{equation}
be the \emph{linear span} of $C*C$.
In fact the set $C*C$ is stable under multiplication by scalars
(because $C$ is), so $C^{\langle 2\rangle}$ can equivalently be defined as
just the \emph{additive span} of $C*C$.

Remark that
the support of $c*c'$ is the intersection of the supports
of $c$ and $c'$. We then call $C^{\langle 2\rangle}$
the \emph{self-intersection span} of $C$.
We will be especially interested in the case $q=2$, where a codeword
can indeed be identified with its support, unambiguously.
Sometimes we will also call $C^{\langle 2\rangle}$ the ``square'' of $C$,
and more generally, higher ``powers'' $C^{\langle t\rangle}$ can be
defined analogously, for any $t\geq0$
(see section \ref{concatenation_et_bilinearite}).

Write $\rate(C)$ and $\drel(C)$ for the rate
and relative minimum distance
of $C$. As a shortcut, write also
$\rated(C)=\rate(C^{\langle 2\rangle})$
and $\ddrel(C)=\drel(C^{\langle 2\rangle})$.
It is easily seen that these functions satisfy:
\begin{equation}
\label{ineg<2>}
\rated\geq\rate\qquad\qquad\ddrel\leq\drel
\end{equation}
(see Proposition~\ref{monotonie} below;
for $q=2$ one even has the stronger result
that $C$ is a subcode of $C^{\langle 2\rangle}$,
since then $c*c=c$ for all $c$).

%Remark that the map $c\mapsto c*c$ from $C$ to $C^{*2}$
%is $1$-to-$1$ or $2$-to-$1$, depending on the characteristic
%NON

Recall that a family of codes $C_i$ of length going to infinity
is said \emph{asymptotically good} if both $\rate(C_i)$ and $\drel(C_i)$
admit a positive asymptotic lower bound.

\begin{theorem}
\label{th1}
For any prime power $q$ (\eg $q=2$), there exists an asymptotically good
family of $q$-ary linear codes $C_i$ whose
self-intersection spans $C_i^{\langle 2\rangle}$ also form an asymptotically
good family. 
\end{theorem}

Keeping \eqref{ineg<2>} in mind, we can rephrase the theorem as
asking for $\epsilon,\epsilon'>0$
such that $\liminf_i\,\rate(C_i)\geq\epsilon$
and $\liminf_i\,\ddrel(C_i)\geq\epsilon'$.
Our proof will be constructive,
for example for $q=2$ we will give an explicit construction
%with $\rho=1/2100$ and $\epsilon=1/700$.
with $\epsilon=1/651$ and $\epsilon'=1/1575$
(more generally all the parameter domain
$\epsilon\leq 0.001872-0.5294\,\epsilon'$
can be attained).

\medskip

Apparently the question of the existence of such codes
was first raised by G.~Z\'emor. The author's interest in it started
from a suggestion of C.~Xing. The generalization
to cubes of codes, or to arbitrarily high powers, is still open
(of course the case of real interest is $q=2$).

While study of the behavior of linear codes
under the operation $*$ is a
very natural problem and certainly deserves investigation for its own sake,
motivation comes as well from applications, such as the
analysis of bilinear algorithms \cite{LW}. There are also links
with secret-sharing and multi-party computation systems \cite{CCCX}\cite{CCX}.
More precisely, suppose given a symmetric $\Fq$-bilinear
map $B:V\times V\longto W$, where $V,W$ are
finite dimensional $\Fq$-vector spaces, as well as
a pair of $\Fq$-linear maps $\phi:V\longto(\Fq)^n$
and $\theta:(\Fq)^n\longto W$, such that the following diagram commutes:
\begin{equation}
\label{diagramme}
\begin{CD}
V\times V @>B>> W \\
@V{\phi\times\phi}VV @AA{\theta}A\\
(\F_q)^n\times(\F_q)^n @>*>> (\F_q)^n
\end{CD}
\end{equation}
that is, such that $B(u,v)=\theta(\phi(u)*\phi(v))$
for all $u,v\in V$.

From the point of view of algebraic complexity theory,
diagram~\eqref{diagramme} expresses how to compute $B$
using only $n$ two-variable multiplications in $\Fq$.
The two maps $\phi,\theta$ are then said to define a 
\emph{(symmetric) bilinear algorithm} of length $n$ for $B$. Of particular
interest is the case where $V=W=\Fqr$ is an extension field
of $\Fq$ and $B$ is usual field multiplication in it:
we refer the reader to
\cite{BP}\cite{CCXY}\cite{ChCh+}
for recent results on this topic.
On the other hand, from the point of view of secret-sharing
and multi-party computation, diagram~\eqref{diagramme} can
be interpreted as follows: elements $u,v\in V$ are split
into shares according to $\phi$ and distributed to $n$ remote
users, these users then multiply their shares locally, and
finally their local results are combined with $\theta$ to
get $B(u,v)$.
Several refinements can then be considered.

First, remark that given finitely many $u^{(i)},v^{(i)}\in V$,
a more general expression such as
\beq
\begin{split}
\sum_i B(u^{(i)},v^{(i)})&=\sum_i \theta(\phi(u^{(i)})*\phi(v^{(i)}))\\
&=\theta(\sum_i \phi(u^{(i)})*\phi(v^{(i)}))
\end{split}
\eeq
can be computed by applying $\theta$ only once, at the very end.
Moreover, letting $C\subset(\Fq)^n$ be the image of $\phi$, we see
that the sum $s=\sum_i \phi(u^{(i)})*\phi(v^{(i)})$ to which $\theta$
is applied at the end of the computation, describes a generic 
element of $C^{\langle 2\rangle}$.
Depending on the context, it could then be desirable 
that this computation be resistant to local alterations
of $s$ caused by noise, or by unreliable users. Also, in a scenario
\`a la threshold cryptography, an important feature will be the
ability to reconstruct $B(u,v)$ knowing only a certain given number
of coordinates of $s$. Clearly, all these properties will be controlled
by the minimum distance of $C^{\langle 2\rangle}$.

\section{Some ideas behind the proof}
\label{informal}

Here we discuss informally some ideas that lead to the proof
of Theorem~\ref{th1}.
Certainly this discussion reflects only the author's own experience
in dealing with this problem.
Since it is not logically necessary for the understanding of the proof,
the reader can skip it with no harm and go directly
to the next section (and maybe come back here later).

\medskip

There is a certain similarity between our object of interest
and the theory of linear intersecting codes \cite{CL}\cite{Miklos}.
%In fact, one of the historical motivations for the introduction
%of intersecting codes was the analysis of bilinear algorithms \cite{LW},
%and it turns out that the study of the coordinatewise multiplication $*$
%acting on linear codes,
%besides having interest by itself,
%also fits well in this framework.
Recall that a linear code $C$ is said intersecting
if $c*c'$ is non-zero for all non-zero $c,c'\in C$
(and this could be refined by requiring $c*c'$ to have at least
a certain prescribed weight).
Although none of these notions imply the other, it turns out
that methods used to produce intersecting codes often produce
codes having a good $\ddrel$. This is often the case, for example,
for intersecting codes constructed as evaluation codes
(see \cite{21sep}\cite{Xing2002} for more on this topic,
although actually
the codes constructed there do not have a good $\ddrel$).

Suppose we are given an algebra $\cF$ of functions, admitting a nice
notion of ``degree'', and which can be evaluated
at a certain set of points $X$. We then define a linear
code $C_D$ as the image of the space $\cF(D)$ of functions
of degree at most $D$ under this evaluation map.
For example, $\cF$ could be the algebra of polynomials
in one or several indeterminates over a finite field, giving rise
to Reed-Solomon or Reed-Muller codes. Or $\cF$ could be the function
field of an algebraic curve, giving rise to Goppa's algebraic-geometry
codes. In all these situations, bounds on the parameters of $C_D$
can be deduced from $D$ and the cardinality of $X$.
Now for $f,f'\in\cF(D)$ we have $ff'\in\cF(2D)$, 
which implies $c*c'\in C_{2D}$ for all $c,c'\in C_D$.
Applying the aforementioned bounds to $C_{2D}$, 
we find that $C_D$ is intersecting
provided $D$ is suitably chosen. But in fact, by linearity,
the argument just above gives
the stronger result $C_D^{\langle 2\rangle}\subset C_{2D}$,
from which the lower bound $\ddrel(C_D)\geq\drel(C_{2D})$
follows.

Remark then that to have a lower bound on $\rate(C_D)$
requires in general $D$ to be large, while a lower bound on $\drel(C_{2D})$
requires $2D$ to be small with respect to the cardinality of $X$.
When the size $q$ of the field is big enough,
these two conditions are compatible: for example,
algebraic-geometry codes verifying the hypotheses in Theorem~\ref{th1}
can be constructed as soon as the Ihara constant satisfies $A(q)>2$
(see sections \ref{AG} and \ref{fin}).
Unfortunately, with the present techniques,
if $q$ is too small,
these two requirements become contradictory when one lets the
length of the codes go to infinity.
A standard solution in such a situation
is to work first over an extension field, and then conclude
with a concatenation argument.
If one is interested only in constructing intersecting codes,
this works easily \cite{21sep}
because a concatenation of intersecting codes is intersecting.
But in the problem we study, things do not behave so nicely:
in general it appears very difficult to derive a lower bound on the
$\ddrel$ of a concatenated code from the parameters of its inner
and outer codes. Perhaps this is best illustrated as follows.

Let $\Fqr$ be an extension of $\Fq$, and let
$\phi:\Fqr\longto(\Fq)^l$ and $\theta:(\Fq)^l\longto\Fqr$
define a multiplication algorithm as discussed in the previous section,
so $xy=\theta(\phi(x)*\phi(y))$ for all $x,y\in\Fqr$.
A very tempting approach when trying to prove Theorem~\ref{th1}
is then to concatenate codes $C$ having asymptotically good squares over
an extension field $\Fqr$, with $\phi$. For if $\phi(C)$ denotes
the concatenated code, it is easily seen that $\theta$ maps
$\phi(C)^{\langle 2\rangle}$ in $C^{\langle 2\rangle}$, hence one could
hope to use this ``reconstruction map'' to derive a lower bound
on the minimum distance of $\phi(C)^{\langle 2\rangle}$ from that
of $C^{\langle 2\rangle}$. More precisely, if $c\in\phi(C)^{\langle 2\rangle}$
has weight less than $d=\dmin(C^{\langle 2\rangle})$, then a fortiori
$c$ has less than $d$ non-zero block symbols over $(\Fq)^l$,
so $\theta(c)\in C^{\langle 2\rangle}$ has weight less than $d$,
hence $\theta(c)=0$. If $\theta$ were injective, we could deduce
that $c=0$. Unfortunately, for $r>1$ it turns out that $\theta$ is never
injective, and all we get is that the block symbols of $c$ all live
in $\ker(\theta)$.
So this ``naive approach'' fails, but not by much: the obstruction
is the kernel of $\theta$.

We fix this as follows. In section \ref{bilinearite_et_extensions}
we define higher ``twisted multiplication laws'' $m_j$ on $\Fqr$,
and we put them together in a map
$\Psi:\Fqr\times\Fqr\longto W$,
where $W=(\Fqr)^{\lceil\frac{r}{2}\rceil}$ if $r$ is odd
(and $W=(\Fqr)^{\frac{r}{2}}\times\mathbb{F}_{q^{\frac{r}{2}}}$ if $r$ is even),
so that:
\begin{itemize}
\item over $\Fq$, $\Psi$ is symmetric bilinear
\item over $\Fqr$, $\Psi$ is a polynomial map of degree $1+q^{\lfloor\frac{r}{2}\rfloor}$.
\end{itemize}
We then construct a bilinear algorithm
\begin{equation}
\label{diagramme2}
\begin{CD}
\Fqr\times\Fqr @>{\Psi}>> W \\
@V{\phi\times\phi}VV @AA{\theta}A\\
(\F_q)^{\frac{r(r+1)}{2}}\times(\F_q)^{\frac{r(r+1)}{2}} @>*>> (\F_q)^{\frac{r(r+1)}{2}}
\end{CD}
\end{equation}
with the property that $\theta$ is \emph{bijective}.
The key steps in proving the bijectivity of $\theta$ are:
\begin{itemize}
\item identify the lower right $(\F_q)^{\frac{r(r+1)}{2}}$ in \eqref{diagramme2}
with the symmetric square $\SqF\Fqr$, that is, with the space through
which any symmetric $\Fq$-bilinear map on $\Fqr$ factorizes uniquely
\item remark that any symmetric $\Fq$-bilinear map on $\Fqr$ can be expressed
uniquely in terms of the $m_j$ for $0\leq j\leq\lfloor\frac{r}{2}\rfloor$,
and $\Fq$-linear operations.
\end{itemize}
We can then concatenate with $\phi$ as in the naive approach above.
In appropriate bases, the matrix of $\phi$, that is, the generating
matrix of the inner code, is made of all $\{0,1\}$ columns of weight $1$ or $2$.
For example, for $r=4$, it would look like
\beq
G_\phi=\left(\begin{array}{cccccccccc}
1 & 0 & 0 & 0 & 1 & 1 & 1 & 0 & 0 & 0 \\
0 & 1 & 0 & 0 & 1 & 0 & 0 & 1 & 1 & 0 \\
0 & 0 & 1 & 0 & 0 & 1 & 0 & 1 & 0 & 1 \\
0 & 0 & 0 & 1 & 0 & 0 & 1 & 0 & 1 & 1
\end{array}\right)
\eeq
although actually (for $q=2$) we will take $r=9$.

Now $\theta$ has no kernel, so we can derive a lower bound
on the minimum distance of the
squared concatenated code $\phi(C)^{\langle 2\rangle}$
by the very same argument as sketched above. This is done
in section~\ref{concatenation_et_bilinearite}.
However there is then an added difficulty:
since $\Psi$ has degree $1+q^{\lfloor\frac{r}{2}\rfloor}$, this bound
will not be in terms of the minimum distance of the square of
the outer code $C$ only,
but also that of its higher powers up to order $1+q^{\lfloor\frac{r}{2}\rfloor}$.

So to conclude
(sections \ref{AG} and \ref{fin})
we need codes over $\Fqr$ whose powers
up to order $1+q^{\lfloor\frac{r}{2}\rfloor}$ are asymptotically good.
On the other hand, algebraic geometry provides codes over $\Fqr$ whose powers
up to order $\lceil A(q^r)\rceil-1$ are asymptotically good.
When $r$ is even, this is not enough: because there we have
$\lfloor\frac{r}{2}\rfloor=\frac{r}{2}$ and we face the
Drinfeld-Vladut bound \cite{DV} $A(q^r)\leq q^{\frac{r}{2}}-1$.
However, when $r$ is odd, we have $\lfloor\frac{r}{2}\rfloor=\frac{r-1}{2}$,
which leaves us just enough room under the Drinfeld-Vladut bound
to make use of a recent
construction \cite{GSBB} of Garcia-Stichtenoth-Bassa-Beelen,
that provides us with curves sufficiently close to it
(although not attaining it) to meet our needs.

\section{Bilinear study of field extensions}
\label{bilinearite_et_extensions}

Let $V$ be a vector space of dimension $r$ over $\Fq$,
and let $V^\vee$ be its dual vector space.
Let also $\Sym(V;\Fq)$ be the space of symmetric bilinear forms on $V$.
If $\lambda\in V^\vee$ is a linear form on $V$,
we can define
\beq
\begin{array}{cccc}
\lambda^{\tens 2}: & V\times V & \longto & \Fq\\
& (u,v) & \mapsto & \lambda(u)\lambda(v)
\end{array}
\eeq
which is a symmetric bilinear form on $V$.

\begin{lemma}
\label{lambda2_base}
Let $\lambda_1,\dots,\lambda_r$ be a basis of $V^\vee$.
Then the $\frac{r(r+1)}{2}$ elements
$\lambda_i^{\tens 2}$ for $1\leq i\leq r$ and $(\lambda_i+\lambda_j)^{\tens 2}$
for $1\leq i<j\leq r$ form a basis of $\Sym(V;\Fq)$.
\end{lemma}
\begin{proof}
Using $\lambda_1,\dots,\lambda_r$ as coordinate functions
we can suppose $V=(\Fq)^r$.
Then $\lambda_i^{\tens 2}$ is the symmetric bilinear form
\begin{equation}
\label{forme_diag}
(u,v)\mapsto u_iv_i
\end{equation}
and $(\lambda_i+\lambda_j)^{\tens 2}$ is $(u,v)\mapsto (u_i+u_j)(v_i+v_j)$,
hence 
$(\lambda_i+\lambda_j)^{\tens 2}-\lambda_i^{\tens 2}-\lambda_j^{\tens 2}$ is
\begin{equation}
\label{forme_croisee}
(u,v)\mapsto u_iv_j+u_jv_i.
\end{equation}
Then we conclude by recognizing these
\eqref{forme_diag} and \eqref{forme_croisee}
as forming the standard basis of $\Sym((\Fq)^r;\Fq)$.
\end{proof}

We will now be interested in the case $V=\Fqr$ is an extension field,
which can indeed be considered as a vector space over $\Fq$,
and we let $\gamma_1,\dots,\gamma_r$ be a basis
(for example $\gamma_i=\gamma^{i-1}$ for some choice of
a primitive element $\gamma\in\Fqr$).
Let also $\Tr:\Fqr\longto\Fq$ denote the trace function.
To each $a\in\Fqr$ we can associate a linear form
\beq
\begin{array}{cccc}
t_a: & \Fqr & \longto & \Fq\\
& x & \mapsto & \Tr(ax).
\end{array}
\eeq
The following is well known:
\begin{lemma}
\label{dualite_trace}
The map
\beq
\begin{array}{ccc}
\Fqr & \longto & (\Fqr)^\vee\\
a & \mapsto & t_a
\end{array}
\eeq
is an isomorphism of $\Fq$-vector spaces.
In particular, $t_{\gamma_1},\dots,t_{\gamma_r}$ form a basis
of $(\Fqr)^\vee$.
\end{lemma}

As a field, $\Fqr$ is endowed with its usual multiplication law,
which we will denote by $m_0$, so
\beq
m_0(x,y)=xy
\eeq
for $x,y\in\Fqr$. For any integer $j\geq1$, we can also define
a ``twisted multiplication law'' $m_j$ by
\beq
m_j(x,y)=xy^{q^j}+x^{q^j}y.
\eeq
Remark that these maps are symmetric and $\Fq$-bilinear
(although not $\Fqr$-bilinear in general). 

\begin{proposition}
\label{deux_bases}
Choose an ordering of the set
$\{t_{\gamma_i}\}_{1\leq i\leq r}\cup\{t_{\gamma_i+\gamma_j}\}_{1\leq i<j\leq r}$
and rename its elements accordingly, say:
\beq
\{t_{\gamma_i}\}_{1\leq i\leq r}\cup\{t_{\gamma_i+\gamma_j}\}_{1\leq i<j\leq r}=\{\phi_1,\dots,\phi_{\frac{r(r+1)}{2}}\}.
\eeq
Then:
\begin{itemize}
\item The family
\beq
(\phi_1^{\tens 2},\dots,\phi_{\frac{r(r+1)}{2}}^{\tens 2})
\eeq
is a basis of $\Sym(\Fqr;\Fq)$.
\item
If $r=2s+1$ is odd, the family
\beq
(t_{\gamma_i}\circ m_j)_{\substack{1\leq i\leq r\\0\leq j\leq s}}
\eeq
is a basis of $\Sym(\Fqr;\Fq)$.
\end{itemize}
\end{proposition}
\begin{proof}
The first claim is a consequence of Lemma~\ref{lambda2_base}
and Lemma~\ref{dualite_trace}.
To prove the second claim, start by remarking that the given
family has the correct size $r(s+1)=\frac{r(r+1)}{2}$.
It suffices thus to show that it is a generating family,
and for this (because of the first claim)
it suffices to show that each $t_a^{\tens 2}$,
for $a\in\Fqr$, can be written as a linear combination
of the $t_b\circ m_j$, for $b\in\Fqr$ and $0\leq j\leq s$.
However for any $x,y\in\Fqr$ we have
\beq
\begin{split}
\Tr(ax)\Tr(ay)&=(ax+a^qx^q+\cdots+a^{q^{2s}}x^{q^{2s}})(ay+a^qy^q+\cdots+a^{q^{2s}}y^{q^{2s}})\\
&=\Tr(a^2xy)+\sum_{1\leq j\leq s}\Tr(a^{1+q^j}(xy^{q^j}+x^{q^j}y))
\end{split}
\eeq
which can be restated 
\beq
t_a^{\tens 2}=\sum_{0\leq j\leq s}t_{a^{1+q^j}}\circ m_j
\eeq
as wanted.
\end{proof}

%\begin{corollary}
%\label{equivalence}
%Keep the previous notations, with $r=2s+1$ odd.
%Then for any family $(\alpha_{x,y})_{x,y\in\Fqr}$ of elements of $\Fq$,
%the two conditions:
%\begin{itemize}
%\item
%$\sum_{x,y\in\Fqr}\alpha_{x,y}\phi_j(x)\phi_j(y)=0$ in $\Fq$,
%for all $1\leq j\leq\frac{r(r+1)}{2}$
%\item
%$\sum_{x,y\in\Fqr}\alpha_{x,y}m_j(x,y)=0$ in $\Fqr$,
%for all $0\leq j\leq s$
%\end{itemize}
%are equivalent.
%\end{corollary}
%\begin{proof}
%Direct consequence of Proposition~\ref{deuxbases}
%(after using Lemma~\ref{dualite_trace} to express the second condition
%in terms of basis coordinates), which shows
%in fact that they are also equivalent to the third condition:
%$\sum_{x,y\in\Fqr}\alpha_{x,y}F(x,y)=0$ for all $F\in\Sym(\Fqr;\Fq)$.
%\end{proof}

From now on we suppose $r=2s+1$
is odd, so $\frac{r(r+1)}{2}=(s+1)(2s+1)$. 
Consider the symmetric $\Fq$-bilinear maps
\beq
\Phi=(\phi_1^{\tens 2},\dots,\phi_{(s+1)(2s+1)}^{\tens 2}):\Fqrs\times\Fqrs\longto(\Fq)^{(s+1)(2s+1)}
\eeq
and
\beq
\Psi=(m_0,\dots,m_s):\Fqrs\times\Fqrs\longto(\Fqrs)^{s+1}.
\eeq
Proposition~\ref{deux_bases} can then be restated as follows:

\begin{corollary}
\label{change_base}
There is an isomorphism
of $\Fq$-vector spaces
\beq
\theta:(\Fq)^{(s+1)(2s+1)}\overset{\sim}{\longto}(\Fqrs)^{s+1}
\eeq
such that
\beq
\theta\circ\Phi=\Psi.
\eeq
%for any family $(\alpha_{x,y})_{x,y\in\Fqr}$ of elements of $\Fq$,
%\beq
%\theta\left(\sum_{x,y\in\Fqr}\alpha_{x,y}\Phi(x,y)\right)=\sum_{x,y\in\Fqr}\alpha_{x,y}\Psi(x,y).
%\eeq
\end{corollary}
\begin{proof}
Set $r=2s+1$,
use the $t_{\gamma_i}$ as coordinate functions on $\Fqr$ as allowed
by Lemma~\ref{dualite_trace},
and define $\theta$ as the invertible linear transformation
that maps the first basis of $\Sym(\Fqr;\Fq)$ given in
Proposition~\ref{deux_bases} to the second one.
\end{proof}

\begin{remark}
\label{carre_symetrique}
For the more sophisticated reader, recall that the symmetric square
of a vector space $V$ over $\Fq$ can be defined, for our purpose,
as the dual of the space of symmetric bilinear forms on it:
$\SqF V=\Sym(V;\Fq)^\vee$.
We let $(u,v)\mapsto u\cdot v$ be the universal
symmetric bilinear map $V\times V\longto\SqF V$,
where $u\cdot v\in\SqF V$ is
the ``evaluation'' element that sends $F\in\Sym(V;\Fq)$ to $F(u,v)$.
Recall also the universal property of the symmetric square:
for any $\Fq$-vector space $W$,
there is a natural identification
\beq
\left\lbrace
\begin{array}{c}
\text{symmetric bilinear maps}\\
V\times V\longto W
\end{array}
\right\rbrace
=
\left\lbrace
\begin{array}{c}
\text{linear maps}\\
\SqF V\longto W
\end{array}
\right\rbrace
\eeq
as $\Fq$-vector spaces,
where a linear map $f:\SqF V\longto W$ corresponds to the
symmetric bilinear map $(u,v)\mapsto f(u\cdot v)$.

So, in the case $V=\Fqrs$,
the symmetric bilinear maps $\Phi$ and $\Psi$
give rise to linear maps $\overline{\Phi}$ and $\overline{\Psi}$ on $\SqF\Fqrs$,
and Proposition~\ref{deux_bases} expresses that these
\beq
\begin{array}{cccc}
\overline{\Phi}: & \SqF\Fqrs & \overset{\sim}{\longto} & (\Fq)^{(s+1)(2s+1)}\\
& x\cdot y & \mapsto & (\phi_1(x)\phi_1(y),\,\phi_2(x)\phi_2(y),\,\dots\,)
\end{array}
\eeq
and
\beq
\begin{array}{cccc}
\overline{\Psi}: & \SqF\Fqrs & \overset{\sim}{\longto} & (\Fqrs)^{s+1}\\
& x\cdot y & \mapsto & (xy,\,xy^q+x^qy,\,\dots,\,xy^{q^s}+x^{q^s}y)
\end{array}
\eeq
are isomorphisms of $\Fq$-vector spaces
(while $\theta=\overline{\Psi}\circ\overline{\Phi}^{-1}$
in Corollary~\ref{change_base}).

A similar result can be given in the case of an even degree
extension $\F_{q^{2s}}$, with only one minor change. Indeed, in this case
remark that one has $(xy^{q^s}+x^{q^s}y)^{q^s}=x^{q^s}y+xy^{q^s}$
for all $x,y\in\F_{q^{2s}}$, which means
that $m_s$ takes values in the subfield $\F_{q^s}$ of $\F_{q^{2s}}$.
Then the very same arguments as before show that $m_0,\dots,m_s$
induce an isomorphism of $\Fq$-vector spaces
\beq
\SqF\F_{q^{2s}}\overset{\sim}{\longto}(\F_{q^{2s}})^s\times\F_{q^s},
\eeq
and composing with traces
gives a basis of $\Sym(\F_{q^{2s}};\Fq)$
in this case also.
\end{remark}

\section{Bilinear study of concatenated codes}
\label{concatenation_et_bilinearite}

If $\cA$ is a vector space
of finite dimension over $\Fq$, if $n\geq1$ is an integer
and $C\subset\cA^n$
is a linear subspace, and if $f:\cA\longto\cB$ is a linear
map from $\cA$ to another vector space $\cB$, we denote
by $f(C)\subset\cB^n$ the subspace obtained by applying $f$
componentwise to the ``codewords'' of $C$:
\beq
f(C)=\{(f(c_1),\dots,f(c_n))\in\cB^n\;|\;c=(c_1,\dots,c_n)\in C\subset\cA^n\}.
\eeq
Also if $C'\subset\cA'^n$ is a code of the same length over another
linear alphabet $\cA'$, and if
$F:\cA\times\cA'\longto\cB$ is a bilinear map,
we denote by $\langle F(C,C')\rangle\subset\cB^n$ the linear span of the
set of elements obtained by applying $F$ componentwise to pairs
of codewords in $C$ and $C'$:
\begin{equation}
\label{<F(C,C')>}
\langle F(C,C')\rangle=\{{\sum}_{\substack{c\in C\;\;\\ c'\in C'}}\alpha_{c,c'}(F(c_1,c'_1),\dots,F(c_n,c'_n))\;|\;\alpha_{c,c'}\in\Fq\}
\end{equation}
which generalizes \eqref{C<2>}.

We will be interested in the case $\cA=\Fqrs$ is an odd degree
extension field of $\Fq$.
Recall the notations from the previous section.
First we have the linear map
\beq
\phi=(\phi_1,\dots,\phi_{(s+1)(2s+1)}):\Fqrs\longto(\Fq)^{(s+1)(2s+1)}
\eeq
as well as the symmetric bilinear map
\beq
\Phi=(\phi_1^{\tens 2},\dots,\phi_{(s+1)(2s+1)}^{\tens 2}):\Fqrs\times\Fqrs\longto(\Fq)^{(s+1)(2s+1)}.
\eeq
%\beq
%\Phi=(\phi_1^{\tens 2},\dots,\phi_{(s+1)(2s+1)}^{\tens 2}):\Fqrs\times\Fqrs\longto(\Fq)^{(s+1)(2s+1)}.
%\eeq
If $C\subset(\Fqrs)^n$ is a linear code of length $n$ over $\Fqrs$,
we will consider $\phi(C)$ and $\langle \Phi(C,C)\rangle$ as codes of length
$N=(s+1)(2s+1)n$ over $\Fq$,
using the natural identification $((\Fq)^{(s+1)(2s+1)})^n=(\Fq)^N$.
Then:
\begin{lemma}
\label{Phi=phi<2>}
With these notations,
\beq
\langle \Phi(C,C)\rangle=\phi(C)^{\langle 2\rangle}.
\eeq
\end{lemma}
\begin{proof}
Direct consequence of the definitions.
\end{proof}

We also have the symmetric $\Fq$-bilinear maps
\beq
m_j:\Fqrs\times\Fqrs\longto\Fqrs
\eeq
for $0\leq j\leq s$, from which we formed
\beq
\Psi=(m_0,\dots,m_s):\Fqrs\times\Fqrs\longto(\Fqrs)^{s+1}.
\eeq
Remark that $m_0$ is not only $\Fq$-bilinear, it is also $\Fqrs$-bilinear.
So if the code $C\subset(\Fqrs)^n$ is $\Fqrs$-linear, then so is $\langle m_0(C,C)\rangle$.
In fact 
$\langle m_0(C,C)\rangle=C^{\langle 2\rangle}$ provided
now componentwise multiplication $*$
is meant over $\Fqrs$.

On the other hand, for $j\geq 1$, $m_j$ is only $\Fq$-bilinear.
So $\langle m_j(C,C)\rangle$ will only be a $\Fq$-linear
subspace of $(\Fqrs)^n$ (and similarly for $\langle \Psi(C,C)\rangle$). 
Nevertheless we will still define the weight of a codeword in $\langle m_j(C,C)\rangle$
and the minimum distance $\dmin(\langle m_j(C,C)\rangle)$
as the usual weight and distance taken in $(\Fqrs)^n$, that is,
over the alphabet $\Fqrs$.

\begin{proposition}
\label{inegalite_dmin}
With the notations above,
\beq
\dmin(\phi(C)^{\langle 2\rangle})\geq\min_{0\leq j\leq s}\dmin(\langle m_j(C,C)\rangle).
\eeq
\end{proposition}
\begin{proof}
%Any codeword in $\phi(C)^{\langle 2\rangle}$
%is of the form
%\beq
%\sum_{c,c'\in C}\alpha_{c,c'}(\phi_j(c_i)\phi_j(c'_i))_{\substack{1\leq i\leq n\quad\quad\\ 1\leq j\leq\frac{r(r+1)}{2}}}
%\eeq
%with $\alpha_{c,c'}\in\Fq$, where we wrote the elements of  
%$(\Fq)^N=(\Fq)^{(n\times\frac{r(r+1)}{2})}$ in matrix form in order to
%keep track of the block structure of the set of indices.
%We have to prove that if such a codeword has weight
%\begin{equation}
%\label{w<dmin}
%w<\dmin(\langle m_j(C,C)\rangle
%\end{equation}
%for all $0\leq j\leq s$,
%then it is the zero codeword.
Let $c\in\phi(C)^{\langle 2\rangle}$ be a codeword.
We have to show that if $c$ has weight
\begin{equation}
\label{w<dmin}
w<\dmin(\langle m_j(C,C)\rangle)
\end{equation}
for all $0\leq j\leq s$,
then it is the zero codeword.

Here $c$ is seen as a word of length $N$ over the alphabet $\Fq$,
but we can also see it as a word of length $n$ over the alphabet
$(\Fq)^{(s+1)(2s+1)}$, and as such obviously it has weight
\begin{equation}
\label{wtilde}
\widetilde{w}\leq w.
\end{equation}
Now, using Corollary~\ref{change_base} and Lemma~\ref{Phi=phi<2>},
we apply $\theta$ blockwise to get a codeword
$\theta(c)\in\langle\Psi(C,C)\rangle$.
Since $\theta$ is invertible, we see that,
considered as a word of length $n$ over the alphabet
$(\Fqrs)^{s+1}$, this $\theta(c)$ has the same weight $\widetilde{w}$.

If we denote by $\pi_0,\dots,\pi_s$ the $s+1$ coordinate projections
$(\Fqrs)^{s+1}\longto\Fqrs$,
then by construction we have $m_j=\pi_j\circ\Psi$, so applying $\pi_j$
blockwise we get a codeword $\pi_j(\theta(c))\in\langle m_j(C,C)\rangle$,
of weight at most $\widetilde{w}$.
But then, $\pi_j(\theta(c))$ is the zero codeword because of
\eqref{w<dmin} and \eqref{wtilde}, and since this holds for all $j$,
we conclude that $\theta(c)$ is zero, hence $c$ is zero.
\end{proof}

\begin{remark}
This is a continuation of Remark~\ref{carre_symetrique}.
Recall from the symmetric square construction
%applied to $\Fqrs$ over $\Fq$
that we have a
universal product $\cdot:\Fqrs\times\Fqrs\longto\SqF\Fqrs$. 
The underlying notion in the proof of Proposition~\ref{inegalite_dmin}
is then that of the ``universal symmetric bilinear span''
\beq
\langle C\cdot C\rangle\subset(\SqF\Fqrs)^n
\eeq
constructed from $C$ and $\cdot$ as in \eqref{<F(C,C')>},
and of which $\langle \Phi(C,C)\rangle=\phi(C)^{\langle 2\rangle}$
and $\langle \Psi(C,C)\rangle$ are two incarnations, under the
invertible linear changes of alphabets $\overline{\Phi}$
and $\overline{\Psi}$. In particular the weight $\widetilde{w}$
in \eqref{wtilde}
should be interpreted as the weight of $c$ considered as a
word over the alphabet $\SqF\Fqrs$.
\end{remark}

\smallskip

Now let $K$ be a finite field (we will apply both cases
$K=\Fq$ and $K=\Fqrs$), and let $*$ denote coordinatewise
multiplication in the vector space $K^n$,
which is a symmetric $K$-bilinear map $K^n\times K^n\longto K^n$.
If $C,C'\subset K^n$ are two linear codes of the same length,
we can define their intersection span
\beq
\langle C*C'\rangle\subset K^n
\eeq
as in \eqref{<F(C,C')>}, and iteratively, setting $C^{<0>}$ as
the $[n,1,n]$ repetition code, we can define
\emph{higher self-intersection spans} (or ``powers'')
\beq
C^{\langle t+1\rangle}=\langle C^{\langle t\rangle}*C\rangle
\eeq
for $t\geq0$.
Equivalently, $C^{\langle t\rangle}$ is the linear span
of the set of coordinatewise products of $t$-uples of codewords
from $C$. 

In particular we have $C^{\langle 1\rangle}=C$,
and $C^{\langle 2\rangle}$ is the same as in \eqref{C<2>}.
More generally we have the natural identities
\beq
\langle C^{\langle t\rangle}*C^{\langle t'\rangle}\rangle=C^{\langle t+t'\rangle}
\eeq
and
\beq
(C^{\langle t\rangle})^{\langle t'\rangle}=C^{\langle tt'\rangle}.
\eeq

\begin{lemma}
\label{coordonnee}
Let $t\geq1$. If $c\in C^{\langle t\rangle}$ is a codeword
and if $i$ is a coordinate
at which $c$ is non-zero, then there is already some $c'\in C$
that is non-zero at $i$.
\end{lemma}
\begin{proof}
Obvious.
\end{proof}

Now given a linear code $C\subset K^n$, for each integer $t\geq0$,
we can define the ``higher'' dimension $\dimt$,
distance $\dtmin$, rate $\ratet$,
and relative distance $\dtrel$ of $C$,
as those parameters for $C^{\langle t\rangle}$.
Then:

\begin{proposition}
\label{monotonie}
Let $C$ be a (non-zero) linear code.
Then for all $t\geq0$, we have
\beq
\dimu(C)\geq\dimt(C)
\eeq
and
\beq
\dumin(C)\leq\dtmin(C).
\eeq
\end{proposition}
\begin{proof}
For $t=0$ these inequalities hold by convention, so we suppose $t\geq1$.
Let $k_t=\dimt(C)$,
and let $S\subset\{1,\dots,n\}$ be an information set of coordinates
for $C^{\langle t\rangle}$.
Without loss of generality we can suppose $S=\{1,\dots,k_t\}$.
Let $G_t$ be the generating matrix of $C^{\langle t\rangle}$
put in systematic form with respect to $S$.
If $c$ is the $i$-th line of $G_t$,
then $c\in C^{\langle t\rangle}$ has a $1$ at coordinate $i$
and is zero over $S\moins\{i\}$.
By Lemma~\ref{coordonnee} we can find $c'\in C$ that is non-zero
at $i$, hence $c*c'\in C^{\langle t+1\rangle}$ is non-zero at $i$
and zero over $S\moins\{i\}$.
Letting $i$ vary we see that $C^{\langle t+1\rangle}$ has full rank
over $S$, hence $\dim C^{\langle t+1\rangle}\geq k_t$. This is
the first inequality.

Now let $d_t=\dtmin(C)$ and let $c\in C^{\langle t\rangle}$ be a codeword
of weight $d_t$. Let $i$ be a non-zero coordinate of $c$, so by
Lemma~\ref{coordonnee} we can find $c'\in C$ that is non-zero
at $i$. Then $c*c'\in C^{\langle t+1\rangle}$ is non-zero at $i$,
so it is not the zero codeword,
and its support is a subset of the support of $c$,
hence $\dmin(C^{\langle t+1\rangle})\leq d_t$. This is
the second inequality.
\end{proof}

\begin{corollary}
\label{sympa}
Let $n\geq k\geq 1$ and $s\geq0$ be integers,
and let $N=(s+1)(2s+1)n$.
Let also $\phi:\Fqrs\longto(\Fq)^{(s+1)(2s+1)}$
be the $\Fq$-linear map defined earlier.
Then, for any $\Fqrs$-linear $[n,k]$ code $C$,
the ``concatenated'' code $\phi(C)$ is a $\Fq$-linear code
of length $N$, and we have:
\begin{enumerate}[(i)]
\item $\quad \dim\phi(C)=(2s+1)\dim C$
\item $\quad \ddmin(\phi(C))\geq\dsmin(C)$
\item $\quad \rate(\phi(C))=\frac{1}{s+1}\rate(C)$
\item $\quad \ddrel(\phi(C))\geq\frac{1}{(s+1)(2s+1)}\dsrel(C)$
\end{enumerate}
where, in the left, parameters (and coordinatewise product) are meant
over $\Fq$, and in the right, they are over $\Fqrs$.
\end{corollary}
\begin{proof}
Remark that $\phi$ is injective, since it was constructed by extending
a basis $t_{\gamma_1},\dots,t_{\gamma_r}$ of $(\Fqr)^\vee$, where $r=2s+1$.
This implies that $\phi(C)$ has dimension $rk$, from which (i) and (iii)
follow.

On the other hand, since $m_0(x,y)=xy$ and $m_j(x,y)=xy^{q^j}+x^{q^j}y$
for $j\geq1$, we find
\beq
\langle m_j(C,C)\rangle\subset C^{\langle 1+q^j\rangle}
\eeq
for all $j\geq0$.
In this inclusion, the right-hand side is a $\Fqrs$-linear code,
while in general the left-hand side is only a $\Fq$-linear subspace.
Nevertheless this implies
\beq
\dmin(\langle m_j(C,C)\rangle)\geq\djmin(C)
\eeq
and together with Propositions~\ref{inegalite_dmin} and~\ref{monotonie},
this gives (ii), and then (iv).
\end{proof}

\section{Algebraic-geometry codes}
\label{AG}

Let $K$ be a finite field.
If $X$ is a (smooth, projective, absolutely irreducible) curve over $K$,
we define a divisor $D$ on $X$ as a formal sum of (closed) points of $X$,
to which one associates the $K$-vector space $L(D)$,
of dimension $l(D)$, made of the functions $f$ on $X$ having
poles at most $D$ (where a pole of negative order means a zero
of opposite order). We also define a degree function on the group
of divisors by extending by linearity the degree function of points.
It is then known:
\begin{itemize}
\item $\; l(D)=0\;$ if $\deg(D)<0$
\item $\; l(D)\geq\deg(D)+1-g\quad$ (Riemann's inequality)
\end{itemize}
where $g$ is the genus of $X$
(and Riemann's inequality can now be seen as part of the
subsequent Riemann-Roch theorem).

If $G=P_1+\cdots+P_n$ is a divisor that is the sum of $n$ distinct
degree~$1$ points of $X$,
then, provided $D$ and $G$ have
disjoint support, we can define an evaluation map
\beq
\begin{array}{cccc}
\ev_{D,G}: & L(D) & \longto & K^n\\
& f & \mapsto & (f(P_1),\dots,f(P_n))
\end{array}
\eeq
and an evaluation code
\beq
C(D,G)\subset K^n
\eeq
as the image of this $K$-linear map $\ev_{D,G}$.
Then, from the preceding properties of $l(D)$ we deduce:
\begin{lemma}[Goppa]
\label{parametres_AG}
Suppose $g\leq\deg(D)<n$.
Then
\beq
\dim C(D,G)=l(D)\geq\deg(D)+1-g
\eeq
and
\beq
\dmin(C(D,G))\geq n-\deg(D).
\eeq
\end{lemma}

Evaluation codes also behave well with regard to our intersection span
operations:

\begin{lemma}
\label{puissances_AG}
For any integer $t\geq0$ we have
\beq
C(D,G)^{\langle t\rangle}\subset C(tD,G).
\eeq
\end{lemma}
\begin{proof}
This is true for $t=0$, so by induction it suffices to show
$\langle C(tD,G)*C(D,G)\rangle\subset C((t+1)D,G)$,
or more generally,
\beq
\langle C(D,G)*C(D',G)\rangle\subset C(D+D',G)
\eeq
for any divisors $D,D'$ with supports disjoint from $G$.
But for $c\in C(D,G)$ and $c'\in C(D',G)$,
write $c=\ev(f)$ and $c'=\ev(f')$ with $f\in L(D)$ and $f'\in L(D')$,
and then $c*c'=\ev(ff')$ with $ff'\in L(D+D')$,
from which the conclusion follows.
\end{proof}

\begin{proposition}
\label{construction_finie}
Let $q$ be a prime power, and $s\geq0$ an integer.
Let $X$ be a curve over $\F_{q^{2s+1}}$, of genus $g$,
and suppose that $X$ admits a set $\{P_1,\dots,P_n\}$
of degree~$1$ points of cardinality
\beq
n>(1+q^s)g.
\eeq
Let then $G=P_1+\cdots+P_n$. Let also $D$ be a divisor on $X$
of support disjoint from $G$ and whose degree $\deg(D)=m$
satisfies
\beq
%g\leq m<(1+q^s)^{-1}n.
g\leq m<\frac{n}{1+q^s}.
\eeq
Finally let
\beq
\phi:\Fqrs\longto(\Fq)^{(s+1)(2s+1)}
\eeq
as in the previous section.
Then the corresponding concatenated code
\beq
C=\phi(C(D,G))\subset(\Fq)^{(s+1)(2s+1)n}
\eeq
has parameters satisfying:
\begin{enumerate}[(i)]
\item $\quad\displaystyle \dim C\geq(2s+1)(m+1-g)$
\item $\quad\displaystyle \ddmin(C)\geq n-(1+q^s)m$
\item $\quad\displaystyle \rate(C)\geq\frac{1}{s+1}\frac{m+1-g}{n}$
\item $\quad\displaystyle \ddrel(C)\geq\frac{1}{(s+1)(2s+1)}\left(1-\frac{(1+q^s)m}{n}\right)$
\end{enumerate}
\end{proposition}
\begin{proof}
Inequalities (i) and (iii) follow
from (i) and (iii) in Corollary~\ref{sympa}
joint with Lemma~\ref{parametres_AG}.

Inequalities (ii) and (iv) follow
from (ii) and (iv) in Corollary~\ref{sympa}
joint with Lemma~\ref{parametres_AG} applied to $C((1+q^s)D,G)$
and Lemma~\ref{puissances_AG} applied with $t=1+q^s$.
\end{proof}

For any prime power $q$, let $N_q(g)$ be the maximal possible number
of degree~$1$ points of a curve of genus $g$ over $\Fq$,
and let
\beq
A(q)=\limsup_{g\to\infty}\frac{N_q(g)}{g}.
\eeq
We will now make use of a recent result
of Garcia-Stichtenoth-Bassa-Beelen \cite{GSBB}, in the following form:

\begin{lemma}
\label{A_grand}
For any prime power $q$, there exists an integer $s$ such that
\beq
A(q^{2s+1})>1+q^s
\eeq
(and in fact this holds as soon as $s$ is large enough).
\end{lemma}
\begin{proof}
If $q$ is a square, one knows from \cite{Ihara}
that $A(q^{2s+1})\geq (q^{2s+1})^{1/2}-1>1+q^s$ as soon as $s$ is large enough.
So suppose $q$ is not a square, say $q=p^{2t+1}$ with $p$ prime.
Then Theorem~1.1 of \cite{GSBB} gives
\beq
A(q^{2s+1})=A(p^{4st+2s+2t+1})\geq\frac{2(p^{2st+s+t+1}-1)}{p+1+\epsilon_s}
\eeq
with $\epsilon_s\to0$ as $s\to\infty$,
so, for $s$ large enough, this is greater
than
\beq
1+q^s=1+p^{2st+s}
\eeq
as claimed.
\end{proof}

From this we can finally prove our main theorem.

\begin{theorem}
\label{construction_asymptotique}
Let $q$ be a prime power, and let $s$ be as given by Lemma~\ref{A_grand}.
Then, for any real number $\mu$ with
\beq
1<\mu<\frac{A(q^{2s+1})}{1+q^s}
\eeq
there exists a family of linear codes $C_i$ over $\Fq$, of length
going to infinity, satisfying
\beq
\liminf_{i\to\infty}\,\rate(C_i)\geq\frac{1}{s+1}\frac{\mu-1}{A(q^{2s+1})}
\eeq
and
\beq
\liminf_{i\to\infty}\,\ddrel(C_i)\geq\frac{1}{(s+1)(2s+1)}\left(1-\frac{(1+q^s)\mu}{A(q^{2s+1})}\right).
\eeq
\end{theorem}
\begin{proof}
For any curve $X$ over $\Fqrs$, denote by $N(X)$ the number
of its degree~$1$ points.
Let $X_i$ be a sequence of curves of genus $g_i$
going to infinity, and such that $\lim_i\frac{N(X_i)}{g_i}=A(q^{2s+1})$.
Also choose a sequence of integers $m_i$ such that $\lim_i\frac{m_i}{g_i}=\mu$.

Now, given $i$ large enough, write $n_i=N(X_i)-1$,
let $P_{i,0},P_{i,1},\dots,P_{i,n_i}$ be the degree~$1$ points of $X_i$,
and let $D_i=m_iP_{i,0}$.
Then Proposition~\ref{construction_finie}
gives a code $C_i$ over $\Fq$ of length $(s+1)(2s+1)n_i$
with $\rate(C_i)\geq\frac{1}{s+1}\frac{m_i+1-g_i}{n_i}$
and $\ddrel(C_i)\geq\frac{1}{(s+1)(2s+1)}\left(1-\frac{(1+q^s)m_i}{n_i}\right)$,
and the conclusion follows.
\end{proof}

Remark that the proof of Theorem~\ref{construction_asymptotique}
is constructive, and works also for a possibly non-optimal sequence
of curves over $\Fqrs$, by which we mean,
curves satisfying $\liminf_i\frac{N(X_i)}{g_i}\geq A'$
for a certain $A'\leq A(q^{2s+1})$,
provided still $A'>1+q^s$ and one replaces all occurences of $A(q^{2s+1})$
in the theorem with $A'$.

%For example \cite{GSBB} gives an explicit sequence of curves over $\F_{2^7}$
%with $\lim_i\frac{N(X_i)}{g_i}\geq A'=105/11\approx 9.545>9=1+2^3$.
%Choosing $\mu=56/55$
%then gives an explicit
%sequence of binary linear codes $C_i$ of length going to infinity
%with $\liminf_i\,\rate(C_i)\geq 1/2100$
%and $\liminf_i\,\ddrel(C_i)\geq 1/700$.
For example \cite{GSBB} gives an explicit sequence of curves over $\F_{2^9}$
with $\lim_i\frac{N(X_i)}{g_i}\geq A'=465/23\approx 20.217>17=1+2^4$.
Choosing $\mu=186/161$
then gives an explicit
sequence of binary linear codes $C_i$ of length going to infinity
with $\liminf_i\,\rate(C_i)\geq 1/651$
and $\liminf_i\,\ddrel(C_i)\geq 1/1575$.
Of course these are only lower bounds, and it could well be that
these codes actually have much better parameters.

\section{Concluding remarks and open problems}
\label{fin}

Keeping Proposition~\ref{monotonie} in mind,
perhaps the most general question one can ask about the parameters
of successive powers of codes is the following:
given a prime power $q$, an integer $n$,
and two sequences $k_1\leq\ k_2\leq k_3\leq\dots$
and $d_1\geq d_2\geq d_3\geq\dots$, does there exist
a linear code $C\subset(\Fq)^n$
with $\dimt(C)=k_t$ and $\dtmin(C)=d_t$ for all $t$?
In fact, already of interest is the study of the function
\beq
a_q^{\langle t\rangle}(n,d)=\max\{k\geq0\,|\,\exists C\subset(\Fq)^n,\,\dim(C)=k,\,\dtmin(C)\geq d\}.
\eeq
%Proposition~\ref{monotonie} gives $a_{q,t}(n,d)\geq a_{q,t+1}(n,d)$,
Proposition~\ref{monotonie} gives $a_q^{\langle t\rangle}(n,d)\geq a_q^{\langle t+1\rangle}(n,d)$,
and Corollary~\ref{sympa} gives
\beq
%a_{q,2}((s+1)(2s+1)n,d)\geq(2s+1)a_{q^{2s+1},1+q^s}(n,d)
a_q^{\langle 2\rangle}((s+1)(2s+1)n,d)\geq(2s+1)a_{q^{2s+1}}^{\langle 1+q^s\rangle}(n,d)
\eeq
for all $s\geq0$.

But besides parameters, one can ask for other characterizations
of codes that are powers. Consider for example the ``square root'' problem:
given a linear code $C\subset(\Fq)^n$, can one decide if there exists
a code $C_0$ such that $C=C_0^{\langle 2\rangle}$, and if so, how many
are there? can one construct one such square root,
or all of them, effectively?

An obvious counting argument shows that, on average, a code taken
randomly in the set of all codes of given length admits one square root.
However the actual distribution of squares within the set of codes
of given parameters might be quite inhomogeneous, and would be interesting
to study. For example, all binary codes of length $3$, except two of
them, are their own unique square root. The two exceptions are:
the $[3,2,2]$ parity code is not a square; the trivial $[3,3,1]$ code
admits two square roots, namely itself and the $[3,2,2]$ code.

\medskip

Now we turn to asymptotic properties.
Define
\beq
\alpha_q^{\langle t\rangle}(\delta)=\limsup_{n\to\infty}\frac{a_q^{\langle t\rangle}(n,\lfloor\delta n\rfloor)}{n},
\eeq
\beq
\delta_q(t)=\sup\{\delta\geq0\,|\,\alpha_q^{\langle t\rangle}(\delta)>0\},
\eeq
and
\beq
\tau(q)=\sup\{t\in\N\,|\,\delta_q(t)>0\}.
\eeq
That is, $\tau(q)$ is the supremum value (possibly $+\infty$)
of $t$ such that there exists an asymptotically good
family of linear codes $C_i$ over $\Fq$
whose $t$-th powers $C_i^{\langle t\rangle}$
also form an asymptotically good family.

From Corollary~\ref{sympa} one finds
\beq
\alpha_q^{\langle 2\rangle}(\delta)\geq\frac{1}{s+1}\alpha_{q^{2s+1}}^{\langle 1+q^s\rangle}((s+1)(2s+1)\delta)
\eeq
and
\beq
\delta_q(2)\geq\frac{1}{(s+1)(2s+1)}\delta_{q^{2s+1}}(1+q^s)
\eeq
for all $s\geq0$.

On the other hand, from Lemma~\ref{parametres_AG}
and Lemma~\ref{puissances_AG} one easily finds
\beq
\alpha_q^{\langle t\rangle}(\delta)\geq\frac{1-\delta}{t}-\frac{1}{A(q)}
\eeq
and
\beq
\delta_q(t)\geq 1-\frac{t}{A(q)}
\eeq
hence
\beq
\tau(q)\geq \lceil A(q)\rceil-1
\eeq
(which is non-trivial only for $q$ large).

Combining these bounds, or equivalently, eliminating $\mu$
from the two estimates in Theorem~\ref{construction_asymptotique},
one gets
\beq
\alpha_q^{\langle 2\rangle}(\delta)\geq\frac{1}{s+1}\left(\frac{1}{1+q^s}-\frac{1}{A(q^{2s+1})}\right)-\frac{2s+1}{1+q^s}\,\delta
\eeq
and
\beq
\delta_q(2)\geq\frac{1}{(s+1)(2s+1)}\left(1-\frac{1+q^s}{A(q^{2s+1})}\right)
\eeq
for all $s\geq0$, and hence, by Lemma~\ref{A_grand},
\beq
\tau(q)\geq 2
\eeq
for all $q$ (which was precisely Theorem~\ref{th1}).

When $q=p$ is prime, these estimates can be made more precise using
the bound
$\frac{1}{A(p^{2s+1})}\leq\frac{1}{2}\left(\frac{1}{p^s-1}+\frac{1}{p^{s+1}-1}\right)$
from \cite{GSBB}.
For $p=2$, the best choice is $s=4$, which gives
\beq
\alpha_2^{\langle 2\rangle}(\delta)\geq\frac{74}{39525}-\frac{9}{17}\,\delta\,\approx\, 0.001872-0.5294\,\delta
\eeq
and
\beq
\delta_2(2)\geq\frac{74}{20925}\approx0.003536.
\eeq
This can be viewed as a quantitative version of the claim $\tau(2)\geq2$
made in the title of this article.
However, in the other direction,
the author doesn't know any upper bound on the $\tau(q)$,
for instance, he doesn't even know whether $\tau(2)$ is finite.

%\section*{Acknowledgment}
%
%The author got interested in this
%problem through an informal discussion with C.~Xing,
%who said he heard it first from G.~Z\'emor.

\end{document}